\documentclass{article}
\usepackage{spconf,amsmath,graphicx}
\usepackage{amssymb}
\usepackage{amsmath}
\usepackage{amsthm}
\usepackage{cite}
\usepackage{bbm}
\usepackage{algpseudocode}
\usepackage{algorithm,algpseudocode}
\usepackage{algorithmicx}
\usepackage{url}
\usepackage{color}

\newtheorem{theorem}{Theorem}
\newtheorem{lemma}{Lemma}

\title{Optimal Design of Energy-Efficient Cell-Free Massive MIMO: \\ Joint Power Allocation and Load Balancing}
\name{Trinh Van Chien, Emil Bj\"ornson, and Erik G. Larsson\thanks{This paper was supported by ELLIIT and CENIIT.}}
\address{Department of Electrical Engineering (ISY), Link\"oping University, SE-581 83 Link\"oping, Sweden\\ {trinh.van.chien, emil.bjornson, erik.g.larsson}@liu.se}

\begin{document}
\ninept
\maketitle
\begin{abstract}
A large-scale distributed antenna system that serves the users by coherent joint transmission is called Cell-free Massive MIMO (multiple input multiple output). For a given user set, only a subset of the access points (APs) is likely needed to satisfy the users' performance demands. To find a flexible and energy-efficient implementation, we minimize the total power consumption at the APs in the downlink, considering both the hardware and transmit powers, where APs can be turned off. Even though this is a non-convex optimization problem, a globally optimal solution is obtained by solving a mixed-integer second-order cone program. We also propose a low-complexity algorithm that exploits group-sparsity in the problem formulation. Numerical results manifest that our optimization framework can greatly reduce the power consumption compared to keeping all APs turned on and only minimizing the transmit powers.
\end{abstract}
\vspace{-2mm}

\section{Introduction}

\vspace{-2mm}

Cell-free Massive MIMO (multiple input multiple output) is a new promising wireless technology to deal with the mediocre cell-edge performance of cellular networks \cite{interdonato2019ubiquitous}.
The system consists of many distributed access points (APs) that transmit coherently in the downlink and process the received signals coherently in the uplink \cite{Ngo2017a, Nayebi2017a}, leading to a higher signal-to-noise ratio (SNR) without using more power. The distributed nature gives massive macro-diversity against pathloss and shadow fading. The spectral efficiency (SE) of Cell-free Massive MIMO has been characterized for Rayleigh fading channels and various implementations in \cite{Nayebi2016a,Ngo2017a, Nayebi2017a,bjornson2019making}. The SE was optimized in \cite{Ngo2017a, Nayebi2017a,ngo2018total,bashar2018enhanced} (among others) to achieve max-min fairness.

In this paper, we instead consider that each user has an SE requirement that the system must satisfy to not cause service interruptions. Hence, the goal of the resource allocation is for the system to deliver the required SEs with as low power consumption as possible, leading to maximum energy efficiency. We stress that this approach to energy efficiency maximization is different from \cite{ngo2018total, nguyen2017energy}. In particular, we will consider the possibility to turn off APs that are not needed to serve the current set of users, bearing in mind that each user will mainly be served by its neighboring APs. This is an important feature since Cell-free Massive MIMO systems may have many APs \cite{Ngo2017a, Nayebi2017a}, where the large number is needed to provide consistent coverage but might not be needed at every time instant.

\subsection{Relation to Prior Work}

The problem of AP activation in distributed antenna systems has previously been analyzed for cloud radio access networks \cite{feng2017boost, vu2015energy, han2016survey}. However, these prior works consider slowly fading channels and perfect channel state information (CSI). In contrast, we consider fast fading channels, imperfect CSI, and pilot contamination. Hence, the optimization problems considered in this paper are entirely different from those in previous work, and more useful for implementation.

 \subsection{Contributions}
 
 In this paper, we consider a Cell-free Massive MIMO system with multi-antenna APs and maximum ratio transmission (MRT) in the downlink. We formulate a total power minimization problem where the active APs and transmit power allocation are the variables. This problem is non-convex but can be solved as a mixed-integer second-order cone (SOC) program. Since this approach is too complex for real-time applications,  we develop a lower-complexity solution by exploiting the inherent sparsity in the problem. The numerical results demonstrate that there are scenarios where only a subset of the APs are needed to satisfy the SE requirements and large power reductions can be achieved by turning off the remaining APs.

\textit{Notations:} We use boldface lower-case and upper-case letters to denote vectors and matrices, respectively. The transpose is denoted by the superscript $(\cdot)^T$ and the Hermitian transpose is denoted by $(\cdot)^H$. The expectation operator is $\mathbb{E} \{ \cdot \}$ and $\mathcal{CN}(\cdot, \cdot)$ denotes a circularly symmetric complex Gaussian distribution. The Euclidean norm, $\ell_1$-norm, and $\ell_{p}$-norm of a vector $\mathbf{x}$ is denoted as $\| \mathbf{x} \|$, $\| \mathbf{x} \|_1$, and $\| \mathbf{x} \|_p$, respectively. The cardinality of a set $\mathcal{X}$ is denoted $|\mathcal{X}|$. 

\section{System Model} \label{Sec:SysModel}
We consider a Cell-free Massive MIMO system with $M$ distributed APs and $K$ single-antenna users.
Each AP has $N$ antennas and is connected to a central processing unit (CPU) via unlimited fronthaul links.
We consider the standard block fading model \cite{massivemimobook}, where the channels are fixed and frequency-flat within a coherence interval of $\tau_c$ channel uses. The channel between AP~$m$ and user~$k$ is modeled by uncorrelated Rayleigh fading as
\begin{equation}
\mathbf{h}_{mk} \sim \mathcal{CN} (\mathbf{0}, \beta_{mk} \mathbf{I}_N),
\end{equation}
where $\beta_{mk} \geq 0$ is the large-scale fading coefficient.

The users are served by coherent joint transmission from the APs.
We assume that each user $k$ has a required downlink SE value $\xi_k>0$ [b/s/Hz] that must be satisfied. The users and APs are arbitrarily distributed, thus it is likely that these SE requirements can be fulfilled without using all the APs. The main goal of this paper is to find a subset $\mathcal{A} \subset \{1,\ldots,M\}$ of active APs  and the corresponding transmit powers that satisfy the SE requirements while minimizing the total power consumption.


We consider a time-division duplex (TDD) protocol where each AP $m$ estimates the channels $\{ \mathbf{h}_{mk} : k=1,\ldots,K\}$ from itself to the users by using uplink pilot transmissions. As in \cite{Ngo2017a, Nayebi2017a}, $\tau_p<K$ orthonormal pilot signals are utilized and assigned to the users. We let $\mathcal{P}_k$ denote the subset of users assigned to the same pilot as user~$k$. By following standard methods for channel estimation and downlink precoding \cite{Chien2016b, interdonato2018}, we obtain the following result.

\begin{lemma} \label{lemma:ClosedFormMRC}
Suppose MRT is used for downlink transmission and AP $m$ allocates power $\rho_{mk'}$ to user $k'$, then an ergodic SE of user~$k$ is
\begin{equation} \label{eq:DLRate}
R_k (\{ \rho_{mk} \}, \mathcal{A}) = \left( 1 - \frac{\tau_p}{\tau_c} \right)\log_2 \left(1 + \mathrm{SINR}_k (\{ \rho_{mk} \}, \mathcal{A}) \right),
\end{equation}
where the effective SINR is
\begin{align}  \notag
&\mathrm{SINR}_k (\{ \rho_{mk} \}, \mathcal{A}) = \\ &\frac{N \left( \sum\limits_{ m \in \mathcal{A}} \sqrt{\rho_{mk} \gamma_{mk}}  \right)^2 }{ N \sum\limits_{k' \in \mathcal{P}_k \setminus \{k\} } \left( \sum_{m \in \mathcal{A}} \sqrt{\rho_{mk'}\gamma_{mk}} \right)^2 + \sum\limits_{k' =1}^K   \sum\limits_{m \in \mathcal{A}} \rho_{mk'} \beta_{mk} +\sigma^2},\label{eq:SINRk}
\end{align}
the noise variance is $\sigma^2$, $p_k$ is the uplink pilot power of user $k$, and the mean-square of the channel estimates between AP $m$ and user $k$ is
\begin{equation}
\gamma_{mk} =  \frac{ \tau_p p_k  \beta_{mk}^2 }{ \tau_p \sum_{k' \in \mathcal{P}_k}  p_{k'}\beta_{mk'} + \sigma^2 }.
\end{equation}
\end{lemma}

The difference between Lemma~\ref{lemma:ClosedFormMRC} and the previous work \cite{Chien2016b, interdonato2018} is that only the subset $\mathcal{A} $ of the APs are active. In \eqref{eq:SINRk}, the numerator is proportional to $N$ which is the array gain achieved by having multiple antennas at each AP. The first term in the denominator is coherent interference from the pilot-sharing users. The second term is conventional non-coherent interference and the third term is noise.

\section{Total Power Minimization Problem} \label{Sec:TotalPowerOpt}

In this section, we formulate a total power minimization problem under the SE requirements of the users and a maximum transmit power $P_{\max}$ per AP. We then rewrite it as a mixed-integer SOC program, which can be solved but with high complexity. 
Similar to \cite{ngo2018total}, we model the total power consumption of the all active APs as
\begin{align}
P_{\mathrm{total}} (\{ \rho_{mk} \}, \mathcal{A}) =& \sum_{m \in \mathcal{A}} \Delta \sum_{k=1}^K  \rho_{mk} + \sum_{m \in \mathcal{A}} P_{\mathrm{act}}
\label{eq:Ptotal}
\end{align}
where the first term is the transmit power multiplied with the inefficiency factor  $\Delta \geq 1$ of the power amplifiers and the second term contains the fixed power $P_{\mathrm{act}}$ per active AP. 

The total power consumption minimization problem is
\begin{subequations} \label{Prob:TotalPowerOpt}
\begin{alignat}{2}
& \underset{\{ \rho_{mk} \geq 0 \}, \mathcal{A} }{\mathrm{minimize}}
&& \, P_{\mathrm{total}} (\{ \rho_{mk} \}, \mathcal{A})  \\
& \,\textrm{subject to}
&&  R_{k} (\{ \rho_{mk} \}, \mathcal{A}) \geq \xi_{k}, \quad \forall k, \label{TotalPowerOpt:b} \\
&&& \sum_{k=1}^K  \rho_{mk}  \leq P_{\max}, \quad\forall m \in \mathcal{A}.
\end{alignat}
\end{subequations}
This is a complicated problem since all APs can transmit data to all users and also cause interference to all users. We first define $\nu_k = 2^{\xi_k \tau_c /(\tau_c - \tau_p)} -1, \forall k$ and notice that the SE constraint in \eqref{TotalPowerOpt:b} can be rewritten as
\begin{equation} \label{eq:SINR-constraints}
 \mathrm{SINR}_{k} (\{ \rho_{mk} \}, \mathcal{A})  \geq \nu_{k}, \quad \forall k.
 \end{equation}
To simplify the problem, we introduce the notation $\mathbf{r}_{\mathcal{A}}  \in \mathbb{C}^{|\mathcal{A}|K + 1} $, $\mathbf{z}_{k\mathcal{A}}  \in \mathbb{C}^{ |\mathcal{A}|}$, $\mathbf{g}_{k\mathcal{A}} \in \mathbb{C}^{ |\mathcal{A}|}$, $\mathbf{U}_{\mathcal{A}} \in \mathbb{C}^{|\mathcal{A}| \times K}$, $\mathbf{s}_{k\mathcal{A}} \in \mathbb{C}^{K + |\mathcal{P}_k|}$. These vectors and matrices are defined as
\begin{align}
\mathbf{r}_{\mathcal{A}} =& \Bigg[ \sqrt{\Delta_{m_{1'}} \rho_{\Delta_{m_{1'}}1}}, \ldots,  \sqrt{\Delta_{m_{|\mathcal{A}|}} \rho_{m_{|\mathcal{A}|}K}}, \sqrt{\sum_{m \in \mathcal{A}} P_{\mathrm{act}}} \Bigg]^T,\label{eq:r}\\
\mathbf{z}_{k\mathcal{A}} =& \left[\sqrt{z_{m_{1'}k}}, \ldots,\sqrt{z_{m_{|\mathcal{A}|}k}} \right]^T, \label{eq:zk} \\
\mathbf{g}_{k\mathcal{A}} =& \left[ \sqrt{N \gamma_{1k}}, \ldots,  \sqrt{N \gamma_{m_{\mathcal{A}}k}} \right]^T ,\label{eq:gk} \\
\mathbf{U}_{\mathcal{A}} =& [\mathbf{u}_1, \ldots, \mathbf{u}_K]^T , \label{eq:U}\\
\mathbf{s}_{k\mathcal{A}} =& \Big[\sqrt{\nu_k}\big( \mathbf{g}_{k\mathcal{A}}^T \mathbf{u}_{t_1'}, \ldots, \mathbf{g}_{k\mathcal{A}}^T \mathbf{u}_{t_{|\mathcal{P}_k \setminus \{k\}|}'}, \ldots \notag \\ & \quad \quad \| \mathbf{z}_{k\mathcal{A}} \circ \mathbf{u}_1 \|, \ldots,  \| \mathbf{z}_{k\mathcal{A}} \circ \mathbf{u}_K \|, \sigma \big) \Big]^T , \label{eq:sk}
\end{align}
 where  $m_{1'},\ldots, m_{|\mathcal{A}|} $ are the members of the set $\mathcal{A}$ (i.e., the indices of the active APs). The $k$th column of $\mathbf{U}_{\mathcal{A}}$ in \eqref{eq:U} is denoted $\mathbf{u}_k = [\sqrt{\rho_{1k}}, \ldots, \sqrt{\rho_{m_{\mathcal{A}k}}}]^T$ and the $m$th row is denoted as $\mathbf{u}_m'$. In \eqref{eq:sk}, $t_1', \ldots, t_{|\mathcal{P}_k \setminus \{k\}|}'$ are the indices of the users belonging to the set $\mathcal{P}_k \setminus \{k\}$, and $|\mathcal{P}_k|$ is the cardinality of the set $\mathcal{P}_k$. The Hadamard product is denoted by $\circ$. By using these notations and \eqref{eq:SINR-constraints}, we can obtain an equivalent epigraph representation of problem~\eqref{Prob:TotalPowerOpt} as
\begin{subequations} \label{Prob:NonConvexSOCP} 
\begin{alignat}{2}
& \underset{ \{ \rho_{mk} \geq 0 \}, \mathcal{A}, s_{\mathcal{A}} }{\mathrm{minimize}} \,\,\,
& &  s_{\mathcal{A}}\\
& \,\,\,\, {\mathrm{subject \,\, to}}
&&  \| \mathbf{r}_{\mathcal{A}} \| \leq  s_{\mathcal{A}}, \label{NonConvexSOCP:b}  \\
& & & \| \mathbf{s}_{k\mathcal{A}} \| \leq \mathbf{g}_{k\mathcal{A}}^T \mathbf{u}_{k\mathcal{A}}, \; \forall k = 1, \ldots, K, \label{NonConvexSOCP:c} \\
& & & || \mathbf{u}_{m}'|| \leq \sqrt{P_{\mathrm{max}}}, \; \forall m \in \mathcal{A}. \label{NonConvexSOCP:d}
\end{alignat}
\end{subequations}
The auxiliary variable $s_{\mathcal{A}}$ moves the objective function  of problem~\eqref{Prob:TotalPowerOpt} to the first constraint in \eqref{NonConvexSOCP:b}. We observe that for a given $\mathcal{A}$, problem~\eqref{Prob:TotalPowerOpt} reduces to an SOC program, as previously shown in \cite{Ngo2017a, Nayebi2017a}. Hence, although \eqref{Prob:NonConvexSOCP} is non-convex and NP-hard, it can be solved by making an exhaustive search over all possible selections of $\mathcal{A}$. Since at least one AP needs to be active if there is $K \geq 1$ users with non-zero SE requirements, there are $2^M - 1$ different selections of the APs that need to be considered in an exhaustive search. 

\subsection{Globally Optimal Solution} \label{subsec:optimal-solution}

Making an exhaustive search to solve \eqref{Prob:NonConvexSOCP} is very computationally costly even in a relatively small network. We will therefore further reformulate the problem to reduce the complexity, while guaranteeing to find the global optimum.

Let the binary optimization variable $\alpha_m \in \{0,1 \}$ characterize the on/off activity of AP $m$. We can then replace the maximum transmit power of AP~$m$ by $\alpha_m^2 P_{\max}$, which takes the original value $P_{\max}$ when the AP is active and is zero when the AP is turned off. This feature is exploited to formulate a mixed-integer SOC program.
\begin{lemma} \label{Theorem:MixedInSOCP}
Consider the mixed-integer SOC program
\begin{subequations} \label{Prob:BnB} 
	\begin{alignat}{2}
	& \underset{ \{ \rho_{mk} \geq 0 \}, \{ \alpha_m \}, s }{\mathrm{minimize}} \,\,\,
	& &  s \\
	& \,\,\,\,\,\,{\mathrm{subject \,\, to}}
	&&  \| \mathbf{r} \| \leq  s, \label{BnB:b}  \\
	& & & \| \mathbf{s}_{k} \| \leq \mathbf{g}_{k}^T \mathbf{u}_{k}, \; \forall k = 1, \ldots, K, \label{BnB:c} \\
	& & & \| \tilde{\mathbf{u}}_{m}'\| \leq \alpha_m \sqrt{P_{\mathrm{max}}}, \; \forall m = 1, \ldots, M, \label{BnB:d} \\
	&&& \alpha_m \in \{0,1 \}, \forall m = 1, \ldots, M,
	\end{alignat}
\end{subequations}
where $\tilde{\mathbf{u}}_{m}'$ is the $m$-th row of matrix $\widetilde{\mathbf{U}} = [\tilde{\mathbf{u}}_1, \ldots, \tilde{\mathbf{u}}_K] \in \mathbb{C}^{M \times K}$ and $\tilde{\mathbf{u}}_k = [\sqrt{\rho_{1k}}, \ldots, \sqrt{\rho_{Mk}}]^T \in \mathbb{C}^M, k = 1\ldots, K$. Moreover, the vectors $\mathbf{r}\in \mathbb{C}^{MK + M}$ and $\mathbf{s}_k \in \mathbb{C}^{K + |\mathcal{P}_k|}$ are defined as
\begin{align}
& \mathbf{r} =  \Big[ \sqrt{\Delta_{1} \rho_{11}}, \ldots,  \sqrt{\Delta_{M} \rho_{MK}}, \alpha_1 \sqrt{P_{\mathrm{act},1}}, \ldots,  \alpha_M \sqrt{P_{\mathrm{act}}} \Big]^T  , \label{eq:rbnb}\\
&\mathbf{s}_k = \Big[\sqrt{\nu_k}\big( \mathbf{g}_{k}^T \mathbf{u}_{t_1'}, \ldots, \mathbf{g}_{k}^T \mathbf{u}_{t_{|\mathcal{P}_k\! \setminus\! \{k\}|}'}, \| \mathbf{z}_{k} \circ \mathbf{u}_1 \|, \ldots, \notag \\ & \quad\quad\quad\quad\quad\quad \| \mathbf{z}_{k} \circ \mathbf{u}_K \|, \sigma \big) \Big]^T ,
\end{align}
with $ \mathbf{z}_{k} = \left[\sqrt{z_{1k}}, \ldots,\sqrt{z_{Mk}} \right]^T $ and $\mathbf{g}_{k} = \left[ \sqrt{N \gamma_{1k}}, \ldots,  \sqrt{N \gamma_{Mk}} \right]^T$.

Problems~\eqref{Prob:NonConvexSOCP} and \eqref{Prob:BnB} are equivalent in the sense that they have the same optimal transmit powers. If we denote by $\{\alpha_{m}^\ast \}$ an optimal solution to the binary variables $\{ \alpha_m \}$ in \eqref{Prob:BnB}, the optimal set of active APs in problem~\eqref{Prob:NonConvexSOCP} is
\begin{equation} \label{eq:ActiveAPv1}
\mathcal{A} = \left\{ m : \alpha_m^\ast = 1,  m \in \{1, \ldots, M \} \right\}.
\end{equation}
\end{lemma}
\begin{proof}
The proof is omitted due to the space limitations.
\end{proof}
Problem~\eqref{Prob:BnB} is a mixed-integer SOC program on standard form and can, thus, be solved by the general-purpose toolbox CVX \cite{cvx2015} using the MOSEK solver \cite{Mosek}. These softwares apply a branch-and-bound approach to deal with the binary variables.

The new binary variables provide the explicit link between the hardware and transmit power consumption, which is an important factor to obtain the global optimum to problem~\eqref{Prob:BnB}. A key reason that we can preserve the SOC structure despite adding the new binary variables is that the binary variables are not involved in the SINR constraints \eqref{BnB:c}. Instead there is an implicit connection via the zero maximum transmit power for inactive APs. This is different from previous approaches (e.g., \cite{feng2017boost}) where $\alpha_m$ appears in the SINRs and therefore would break the SOC structure.

\section{Sparsity-Based Low-Complexity Algorithm} \label{Sec:SubSols}

Motivated by the high computational complexity of solving the total power minimization problem using Lemma~\ref{Theorem:MixedInSOCP}, we will now propose an algorithm that finds a good suboptimal solution with tolerable complexity. Since we are searching for a solution where many of the power variables are zero, we will re-express \eqref{Prob:TotalPowerOpt} as a sparse reconstruction problem where we try to push many of the transmit power variables to become zero. We begin with the following result.

\begin{lemma} \label{lemma:mixedL2L0}
	The original problem~\eqref{Prob:TotalPowerOpt} has the same optimal transmit powers as the following problem
	\begin{equation} \label{Prob:TotalPowerOptLp}
	\begin{aligned}
	& \underset{\{ \rho_{mk} \geq 0 \} }{\mathrm{minimize}}
	&&   \sum_{ m=1}^M \Delta \| \pmb{\rho}_{m} \|^{2}  + \mathbbm{1}_m (\pmb{\rho}_{m}) P_{\mathrm{act}}\\
	& \,\,\textrm{subject to}
	& &  \| \mathbf{s}_{k\mathcal{A}_M} \| \leq \mathbf{g}_{k\mathcal{A}_M}^T \mathbf{u}_{k\mathcal{A}_M}, \; \forall k = 1,\ldots, K,\\
	& & & || \mathbf{u}_{m}'|| \leq \mathbbm{1}_m (\pmb{\rho}_{m}) \sqrt{P_{\mathrm{max}}}, \; \forall m = 1, \ldots, M,
	\end{aligned}
	\end{equation}
	where $\pmb{\rho}_m = [\sqrt{\rho_{m1}}, \ldots, \sqrt{\rho_{mK}}]^T \in \mathbb{C}^K$, $\mathcal{A}_M = \{ 1, \ldots, M\}$,	
	 and
	\begin{equation} \label{eq:IndxFunc}
	\mathbbm{1}_m (\pmb{\rho}_{m}) = \begin{cases}
	1, & \mbox{if } \| \pmb{\rho}_{m} \| > 0, \\
	0, & \mbox{if } \| \pmb{\rho}_{m} \| = 0.
	\end{cases}
	\end{equation}
Moreover, if we denote by $\{ \rho_{mk}^{\ast} \}$ the optimal set of all transmit powers to \eqref{Prob:TotalPowerOptLp}, then the set
\begin{equation} \label{eq:AsetV1}
\mathcal{A} = \left\{ m : \| \pmb{\rho}_m^{\ast} \| = 0, m \in \{1, \ldots, M \} \right\}
\end{equation}
is the optimal set of active APs to problem~\eqref{Prob:TotalPowerOpt}.
\end{lemma}
\begin{proof}
The proof is omitted due to the space limitations. 
\end{proof}
Lemma~\ref{lemma:mixedL2L0} shows that we do not need separate variables for optimizing the active APs set~$\mathcal{A}$, but we can implicitly determine if AP $m$ is active or not by checking if $\| \pmb{\rho}_m \| >0$ or $\| \pmb{\rho}_m \| =0$. The reformulated problem~\eqref{Prob:TotalPowerOptLp} has fewer variables but remains non-convex due to the $\ell_0$-norm in the objective.
A standard approach is to relax the $\ell_0$-norm to a convex $\ell_1$-norm, but we will take another approach that is known to utilize the sparsity more effectively \cite{zhang2015survey}. We consider an $\ell_p$-norm relaxation of problem~\eqref{Prob:TotalPowerOptLp} for some $0<p<1$:\footnote{From the range of the considered $\ell_p$-norms, the condition $0 < \tilde{p}/2 < 1$ as in \eqref{Prob:TotalPowerOptLpv0} leads to $0 < \tilde{p} < 2$.} 
\begin{subequations} \label{Prob:TotalPowerOptLpv0}
\begin{alignat}{2}
&\,\, \underset{\{ \rho_{mk} \geq 0 \} }{\mathrm{minimize}}
&&   \left( \sum_{ m=1}^M \left(\Delta^{2/\tilde{p}} \| \pmb{\rho}_{m} \|^{2} \right)^{\tilde{p}/2}  + P_{\mathrm{act}}^{\tilde{p}/2} \right)^{2/\tilde{p}} \label{TotalPowerOptLpv0:a}\\
& \,\,\mathrm{subject \,\, to}
& &  \| \mathbf{s}_{k \mathcal{A}_M} \| \leq \mathbf{g}_{k\mathcal{A}_M}^T \mathbf{u}_{k\mathcal{A}_M}, \; \forall k = 1,\ldots, K,\\
& & & || \mathbf{u}_{m}'|| \leq \sqrt{P_{\mathrm{max}}}, \; \forall m = 1, \ldots, M.
\end{alignat}
\end{subequations}
The objective function of problem~\eqref{Prob:TotalPowerOptLpv0} treats every vector $\pmb{\rho}_{m}$ as an entity in $\Delta^{2/\tilde{p}} \| \pmb{\rho}_{m} \|^{2}$ when seeking for a sparse solution. This group-sparse approach differs from the previous works that consider element-based  \cite{candes2008enhancing} or beamforming-vector-based  \cite{luo2014downlink} sparsity.

Even though  problem~\eqref{Prob:TotalPowerOptLpv0} remains non-convex after the norm relaxation, we can find a stationary point by adapting the iteratively reweighted least square approach \cite{ba2014}, that was originally developed for component-wise sparsity. By dropping the exponent $2/\tilde{p}$ and $P_{\mathrm{act}}$ in \eqref{TotalPowerOptLpv0:a}, we achieve the following result.

\begin{theorem} \label{Theorem:Sparsity}
	Suppose problem~\eqref{Prob:TotalPowerOptLpv0} is feasible. Since the feasible set is convex, we can construct an iterative algorithm that starts with the given initial weight values $a_m^{(0)} = 1, \forall m=1, \ldots, M,$ and in iteration $n=1,2,\ldots$ solves the following SOC program: \vspace{-2mm}
	\begin{equation} \label{Prob:TotalPowerOptv3}
	\begin{aligned}
	& \underset{\{ \rho_{mk} \geq 0 \} }{\mathrm{minimize}}
	&&   \sum_{ m=1}^M a_m^{(n-1)} \| \pmb{\rho}_{m} \|^{2}\\
	& \,\,\mathrm{subject \,to}
	& &  \| \mathbf{s}_{k \mathcal{A}_M} \| \leq \mathbf{g}_{k\mathcal{A}_M}^T \mathbf{u}_{k\mathcal{A}_M}, \; \forall k = 1,\ldots, K,\\
	& & & || \mathbf{u}_{m}'|| \leq \sqrt{P_{\mathrm{max}}}, \; \forall m = 1, \ldots, M,
	\end{aligned}
	\end{equation}
	to yield the solution $\{ \pmb{\rho}_{m}^{\ast, (n)} \}$, for which
	\begin{equation}
	\pmb{\rho}_{m}^{\ast,(n)} = \left[\sqrt{\rho_{m1}^{\ast,(n)}}, \ldots, \sqrt{\rho_{mK}^{\ast,(n)}} \right]^T \in \mathbb{C}^K,
	\end{equation}
	is the optimal transmit powers for AP~$m$ at iteration~$n$. After that, the weight values are updated for the next iteration as
	\begin{equation} \label{eq:weightv1}
	a_m^{(n)} = \frac{\Delta \tilde{p}}{2}  \left( \| \pmb{\rho}_{m}^{\ast, (n)} \|^2 + \epsilon_n^2 \right)^{ \frac{\tilde{p}}{2} - 1},
	\end{equation}
where $\epsilon_n$ is a sufficient small positive damping constant with $\epsilon_n \leq \epsilon_{n-1}$ and $\lim_{n \rightarrow \infty } \epsilon_n = 0$. The proposed iterative process exhibits the properties below:
	\begin{enumerate}
		\item The objective function of~\eqref{Prob:TotalPowerOptLpv0} reduces after each iteration until reaching a fixed point, which is a stationary point of~\eqref{Prob:TotalPowerOptLpv0}.
		\item If an arbitrary AP~$m$ has zero transmit power at the optimum of iteration $n$, this AP will have zero transmit power in all the following iterations.
	\end{enumerate} 
\end{theorem}
\begin{proof}
The proof is omitted due to the space limitations.
\end{proof}

Theorem~\ref{Theorem:Sparsity} guarantees a monotonically decreasing objective function and the main computational cost is to solve \eqref{Prob:TotalPowerOptv3} in each iteration. The iterative process reaches a stationary point to problem~\eqref{Prob:TotalPowerOptLpv0} where APs are turned off along the iterations.
The damping constant $\epsilon_n >0$ is introduced to cope with numerical issues that can appear when updating the weight values \eqref{eq:weightv1}. The stopping criterion can be selected by comparing two consecutive iterations.

Due to the norm-relaxation, we will not use the solution from Theorem~\ref{Theorem:Sparsity} as the final solution but instead as an indication of which APs to turn off. More precisely, we compute the transmit power that the APs utilize at the solution from Theorem~\ref{Theorem:Sparsity} and reorder the APs in increasing power order.  We then make a bisection-like search over how many APs should be turned on and always keep the ones that utilize the most power at the solution from Theorem~\ref{Theorem:Sparsity}. We thereby identify an AP set that minimizes the total power consumption.

\begin{figure}[t]
		\centering
		\includegraphics[trim=3.1cm 8.0cm 3.5cm 8.5cm, clip=true, width=3.2in]{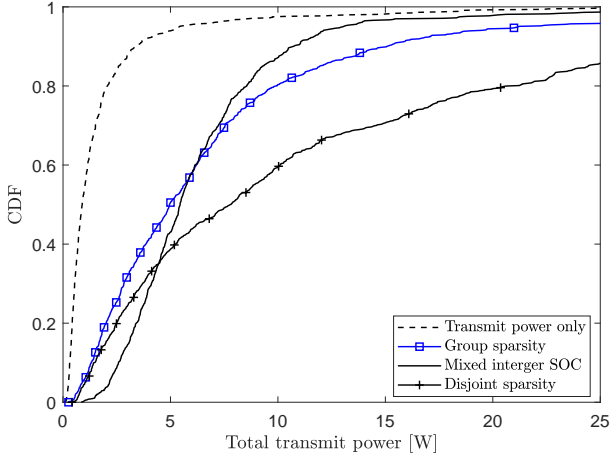} \vspace*{-0.25cm}
		\caption{The CDF of the total transmit power.}
		\label{FigTransmitPowerMRT}
		\vspace*{-0.4cm}
\end{figure}
\begin{figure}[t]
		\centering
		\includegraphics[trim=3.1cm 8.0cm 3.5cm 8.5cm, clip=true, width=3.2in]{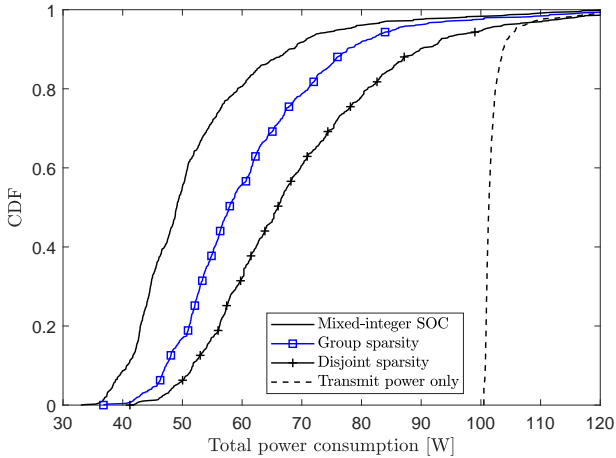} \vspace*{-0.25cm}
		\caption{The CDF of the total power consumption.}
		\label{FigTotalPowerMRT}
		\vspace*{-0.4cm}
\end{figure}

\section{Numerical Results} \label{Sec:NumRes}

We have simulated a setup where $M=20$ APs and $K=20$ users are randomly distributed within a squared area of $1$~km$^2$ by a uniform distribution, given that the distance between two APs should be larger than $50$\,m. Each AP is equipped with $N=20$ antennas. The requested SE from each user is $2$ [b/s/Hz]. We apply the wrap-around structure to get rid of edge effects and guarantee uniform simulation performance for the $M$ APs. Coherence intervals have $\tau_c = 200$ symbols. There are $\tau_p=5$ orthogonal pilots and each is assigned to $4$ randomly selected users. The pilots are transmitted with uplink power $p_k=0.2$~W. We use the large-scale fading model with correlated shadow fading from \cite{bjornson2019making}, which matches well with the 3GPP Urban Microcell model for a carrier frequency $2$~GHz. The maximum downlink power per AP is $P_{\max}=1$~W and the noise variance is $\sigma^2= -92$~dBm.
The power consumption is modeled similar to  \cite{ngo2018total}: $\Delta = 2.5$ and $P_{\mathrm{act}} = 5.03$~W.

We compare the following methods:
\begin{itemize}
\item[$(i)$] \textit{Transmit power only:} All APs are turned on and the total transmit power is minimized  as in \cite{ngo2018total, nguyen2017energy}.
\item[$(ii)$] \textit{Group sparsity:} This is the proposed sparsity-based method from Sec.~\ref{Sec:SubSols} with the norm $\tilde{p}/2 = 0.5$.
\item[$(iii)$] \textit{Disjoint sparsity:} This is a previous method from \cite{shi2016smoothed, luo2014downlink} where the AP subset and transmit power are optimized disjointly.
\item[$(iv)$] \textit{Mixed-integer SOC:} This is the optimal method from Sec.~\ref{subsec:optimal-solution}.
\end{itemize}

Fig.~\ref{FigTransmitPowerMRT} shows the cumulative distribution function (CDF) of the total transmit power [W] achieved by the four different methods, where the randomness is due to the different AP and user locations. When all APs are active, the transmit power is $1.8$~W on the average. The mixed-integer SOC program uses roughly $3.6\times$ more transmit power: $6.4$~W on the average. The proposed sparsity-based method utilizes about $7$~W, while the disjoint sparsity benchmark uses the highest transmit power level: about $11.8$~W.

The proposed methods are not minimizing the transmit power but the total transmit power. 
Fig.~\ref{FigTotalPowerMRT} shows the CDF of the total power consumption [W]. Contrary to the previous figure, the benchmark where all APs are active has the the highest total power consumption: about $102$~W on average. By solving the proposed mixed-integer SOC program, we find the globally minimum total power consumption, which saves about $49\%$ compared with the baseline. The proposed sparsity-based method requires around $17\%$ more power than the global minimum. In contrast, the previous disjoint sparsity benchmark requires $32\%$ more power than the global minimum.

\section{Conclusion} \label{Sec:Conclusion}
This paper has minimized the total downlink power consumption in cell-free Massive MIMO systems by jointly optimizing the number of active APs and their transmit powers, while satisfying the SEs requested by the users.
A globally optimal solution can be found by formulating the considered problem as a mixed-integer SOC program. We observe a considerable reduction in the total power consumption (about $50\%$) compared with only minimizing the total transmit powers as in previous work. Since the mixed-integer SOC program has high computational complexity, we also developed a lower-complexity algorithm that finds a good suboptimal solution with relatively low complexity by utilizing sparsity methods. In the simulation part, this method only requires 20\% more power than the global minimum.

\bibliographystyle{IEEEbib}
\bibliography{IEEEabrv,refs}

\end{document}